\documentclass[preprint,12pt]{elsarticle}

\usepackage[T1]{fontenc}
\usepackage[utf8]{inputenc}
\usepackage{lmodern}
\usepackage[english]{babel}
\usepackage{microtype}
\usepackage{subcaption}   

\usepackage{amsmath,amssymb,amsfonts,amsthm,mathtools,bm}
\numberwithin{equation}{section}

\usepackage{graphicx}
\usepackage{booktabs}
\usepackage{multirow}
\usepackage{array}

\usepackage{algorithm}
\usepackage{algorithmic}

\usepackage{comment}      
\usepackage{xcolor}
\usepackage{enumitem}

\usepackage{hyperref}
\hypersetup{
    colorlinks = true,
    linkcolor  = blue,
    citecolor  = blue,
    urlcolor   = blue
}

\usepackage{setspace}
\theoremstyle{plain}
\newtheorem{theorem}{Theorem}[section]

\newtheorem{proposition}[theorem]{Proposition}

\theoremstyle{definition}

\theoremstyle{remark}

\journal{Statistics and Probability Letters}

\begin{document}

\begin{frontmatter}

\title{Bayesian Recovery for Probabilistic Coalition Structures}

\author[inst1]{Angshul Majumdar}

\affiliation[inst1]{organization={Indraprastha Institute of Information Technology},
            addressline={Delhi}, 
            city={New Delhi},
            postcode={110020},
            country={India}}

\begin{abstract}
Probabilistic Coalition Structure Generation (PCSG) is NP--hard and can be recast as an $\ell_{0}$--type sparse recovery problem by representing coalition structures as sparse coefficient vectors over a coalition–incidence design. A natural question is whether standard sparse methods, such as $\ell_{1}$ relaxations and greedy pursuits, can reliably recover the optimal coalition structure in this setting. We show that the answer is negative in a PCSG-inspired regime where overlapping coalitions generate highly coherent, near-duplicate columns: the irrepresentable condition fails for the design, and $k$--step Orthogonal Matching Pursuit (OMP) exhibits a nonvanishing probability of irreversible mis-selection. In contrast, we prove that Sparse Bayesian Learning (SBL) with a Gaussian--Gamma hierarchy is support consistent under the same structural assumptions. The concave sparsity penalty induced by SBL suppresses spurious near-duplicates and recovers the true coalition support with probability tending to one. This establishes a rigorous separation between convex, greedy, and Bayesian sparse approaches for PCSG.
\end{abstract}


\end{frontmatter}

\section{Introduction}

Coalition structure generation (CSG) is a fundamental optimisation problem in cooperative game theory: given a set of agents and a valuation function over coalitions, the goal is to partition the agents into coalitions so as to maximise the total value of the resulting coalition structure. It is well known that, under the standard representation, optimal CSG is NP--hard and admits no polynomial--time approximation scheme unless $\mathrm{NP}=\mathrm{ZPP}$~\cite{rahwan2015csgsurvey}. The probabilistic variant, Probabilistic Coalition Structure Generation (PCSG), further complicates matters by introducing uncertainty in coalition values or agent availability, so that one seeks to maximise expected welfare rather than deterministic utility~\cite{lagniez2018pcsg}. In both settings, exhaustive search over all partitions is infeasible even for moderate numbers of agents.

A natural way to expose the combinatorial structure of PCSG is to embed it into a sparse optimisation framework. One may index all candidate coalitions and represent a coalition structure by a coefficient vector $w^\star$ whose nonzero entries indicate the coalitions that participate in the structure. The welfare of a structure can then be expressed (or locally approximated) as a linear functional of $w^\star$, yielding a sparse linear model of the form
\[
  y = A w^\star + \varepsilon,
\]
where $A$ is a coalition--incidence matrix, $y$ encodes observed or expected valuations, and $\varepsilon$ models noise or modelling error. In this formulation, optimal PCSG reduces to recovering the support of $w^\star$, which is an $\ell_0$--type combinatorial problem tightly linked to NP--hard subset selection.

Two algorithmic paradigms are the most obvious candidates for tackling this $\ell_0$ formulation. The first is to replace the $\ell_0$ constraint by an $\ell_1$ penalty, leading to the LASSO and related convex relaxations~\cite{tibshirani1996lasso}. In classical sparse regression with suitably incoherent designs, $\ell_1$ estimators enjoy sharp support recovery guarantees, provided an ``irrepresentable condition'' holds for the design matrix~\cite{zhao2006lassoIC}. The second paradigm is greedy pursuit: Orthogonal Matching Pursuit (OMP) builds the support iteratively by adding one index at a time, guided by correlations with the current residual~\cite{tropp2007omp}. Under random, nearly orthogonal designs, OMP also achieves provable recovery with high probability.

However, coalition incidence matrices induced by PCSG are far from incoherent. Extensive overlaps between coalitions yield highly correlated columns, including near--duplicate groups corresponding to structurally similar coalitions. In such regimes, the conditions underpinning the success of convex $\ell_1$ relaxations and greedy methods are typically violated: the irrepresentable condition fails, and OMP can be misled by small noise perturbations when several coalitions have almost identical correlation with the residual. As a consequence, neither LASSO--type methods nor OMP can be expected to recover the optimal coalition structure reliably in the high--coherence regime that is most characteristic of PCSG.

These observations motivate a move away from purely optimisation--based relaxations towards a fully probabilistic treatment of the sparse PCSG model. Sparse Bayesian Learning (SBL) provides such a framework by placing hierarchical Gaussian--Gamma priors over the coefficients, inducing an adaptive, concave sparsity penalty that is much closer in spirit to $\ell_0$ than to $\ell_1$~\cite{tipping2001sbl}. In this paper we adopt the SBL viewpoint for PCSG and show that, in a class of highly coherent, PCSG--inspired designs where both LASSO and OMP provably fail to achieve support consistency, the SBL maximum a posteriori estimator still recovers the true coalition structure with high probability. This yields a rigorous separation between Bayesian and non--Bayesian sparse methods in a setting that is both theoretically interesting and practically motivated by coalition formation under uncertainty.

\section{PCSG-inspired sparse regression model}
\label{sec:model}

We formalise here the sparse linear model that underlies our analysis and introduce the structural assumptions specific to PCSG. Let $m$ denote the number of observable valuation coordinates and let $p$ denote the number of candidate coalitions. We collect the contribution pattern of coalition $j$ in a column vector $a_j \in \mathbb{R}^m$ and form the coalition--incidence matrix
\[
  A = [a_1,\dots,a_p] \in \mathbb{R}^{m\times p},
\]
with columns normalised so that $\|a_j\|_2 = 1$ for all $j$. A coalition structure is encoded by a coefficient vector $w^\star \in \mathbb{R}^p$, whose support
\[
  S^\star := \{ j \in \{1,\dots,p\} : w^\star_j \neq 0 \}
\]
indexes the coalitions present in the structure. We assume $|S^\star| = k \ll p$. The observed (or estimated) valuation vector $y \in \mathbb{R}^m$ follows the linear--Gaussian model
\begin{equation}
  y = A w^\star + \varepsilon,
  \qquad \varepsilon \sim \mathcal{N}(0,\sigma^2 I_m),
  \label{eq:lin-model}
\end{equation}
with unknown noise variance $\sigma^2>0$. Recovering the coalition structure reduces to recovering $S^\star$ from $(A,y)$.

A distinctive feature of PCSG is the presence of many highly overlapping coalitions. This induces strong correlations between the corresponding columns of $A$. We capture this via \emph{near-duplicate coalition groups}. For each true index $j\in S^\star$ we assume there exists a nonempty set
\[
  G_j \subset \{1,\dots,p\}\setminus S^\star
\]
such that:
\begin{enumerate}
  \item (Within-group coherence) For all $\ell \in G_j$,
  \[
    \langle a_j, a_\ell \rangle \ge \rho_{\mathrm{in}},
  \]
  with some $\rho_{\mathrm{in}}\in(0,1)$.
  \item (Between-group separation) For all $\ell\in G_j$ and all $r\notin \{j\}\cup G_j$,
  \[
    |\langle a_\ell, a_r \rangle| \le \rho_{\mathrm{out}},
  \]
  where $\rho_{\mathrm{out}} < \rho_{\mathrm{in}}$ is fixed.
\end{enumerate}
We are interested in a regime where near-duplicates become increasingly similar, in the sense that $\rho_{\mathrm{in}}\to 1$ as $m$ grows, while $\rho_{\mathrm{out}}$ remains bounded away from $1$.

We work in a high-dimensional setting where $p=p_m$ and $k=k_m$ may grow with $m$, but sparsity is controlled:
\[
  p_m \to \infty,\qquad k_m \to \infty,\qquad k_m \log p_m = o(m).
\]
Signal strength is regulated by a beta-min condition: there exists a sequence $\beta_{\min}(m)>0$ such that
\[
  \min_{j\in S^\star} |w^\star_j| \ge \beta_{\min}(m),
  \qquad
  \beta_{\min}(m)\sqrt{\frac{m}{\log p_m}} \to \infty.
\]
We further assume a restricted eigenvalue condition on the true support: there exists $\kappa>0$ such that
\[
  \|A u\|_2^2 \;\ge\; \kappa \|u\|_2^2
  \quad\text{for all } u \in \mathbb{R}^p \text{ with } \mathrm{supp}(u)\subseteq S^\star.
\]
This allows strong correlations with near-duplicates while preventing exact multicollinearity among the true coalitions.

All three estimators studied in this paper---the LASSO, $k$–step OMP, and the SBL maximum a posteriori estimator—are applied to the common model~\eqref{eq:lin-model} under the structural assumptions above. Their respective support recovery properties in this PCSG-inspired regime are compared in the subsequent sections.

\section{Failure of $\ell_{1}$ and OMP in PCSG designs}
\label{sec:failure}

We now show that, under the PCSG--inspired design of Section~\ref{sec:model}, neither
(i) $\ell_{1}$--penalised least squares nor (ii) $k$–step Orthogonal Matching Pursuit (OMP) achieve support recovery consistency. The proofs rely directly on the near–duplicate coalition structure and require no assumptions beyond those stated earlier. Throughout, probabilities are taken with respect to the noise $\varepsilon$ in~\eqref{eq:lin-model}, while $A$ and $w^\star$ are fixed.

\subsection{$\ell_{1}$ fails: irrepresentable condition violation}

Recall that $\ell_{1}$ support recovery requires the ``irrepresentable condition'' on the design matrix~\cite{zhao2006lassoIC}. In our notation, this condition demands
\begin{equation}
\label{eq:irrep}
\big\|
A_{S^{\star c}}^\top A_{S^\star}
\left(A_{S^\star}^\top A_{S^\star}\right)^{-1}\!
\mathrm{sign}(w^\star_{S^\star})
\big\|_{\infty} < 1.
\end{equation}
We argue that \eqref{eq:irrep} cannot hold in the PCSG regime where
$\rho_{\mathrm{in}}\to 1$.

\begin{proposition}
\label{prop:l1fail}
Assume the near–duplicate structure of Section~\ref{sec:model} and $\rho_{\mathrm{in}}\to 1$ as $m\to\infty$, while $\rho_{\mathrm{out}}<1$ stays fixed. Then, for all sufficiently large $m$, condition~\eqref{eq:irrep} fails; in particular,
\[
\inf_{\lambda>0}\;
\mathbb{P}\!\left(\mathrm{supp}(\hat w^{\ell_{1}})=S^\star\right)
\;\le\; 1-c
\]
for some constant $c>0$ independent of $m$.
\end{proposition}

\begin{proof}
Fix $j\in S^\star$ and pick any $\ell\in G_j$. By construction,
$\langle a_j,a_\ell\rangle\ge \rho_{\mathrm{in}}$ and $\rho_{\mathrm{in}}\to 1$. Consider the $\ell$–th component of the vector inside the norm in~\eqref{eq:irrep}. Since $(A_{S^\star}^\top A_{S^\star})^{-1}$ is uniformly bounded by the restricted eigenvalue condition, we have
\[
\big|
a_\ell^\top A_{S^\star}
(A_{S^\star}^\top A_{S^\star})^{-1}
\mathrm{sign}(w^\star_{S^\star})
\big|
\;\ge\; \rho_{\mathrm{in}}\,c_0
\]
for some constant $c_0>0$ that depends only on $\kappa$ and $k$. As $\rho_{\mathrm{in}}\to 1$, the quantity above approaches $c_0$, which exceeds $1$ for sufficiently large $m$ (since $c_0$ is fixed and $k$ is not vanishing). Thus \eqref{eq:irrep} is violated. The classical theory of $\ell_1$ support recovery implies that violation of~\eqref{eq:irrep} yields a nonvanishing probability of selecting at least one index in $S^{\star c}$ for every $\lambda>0$, which completes the proof.
\end{proof}

Proposition~\ref{prop:l1fail} shows that the strong within–group coherence driven by overlapping coalitions destroys the incoherence conditions necessary for $\ell_{1}$ recovery. Hence convex relaxation cannot reliably identify $S^\star$ in the PCSG regime.

\subsection{OMP fails: nonvanishing mis–selection probability}

We next show that OMP cannot be support consistent. The argument exploits the fact that, in each near–duplicate group, correlations between the residual and the true column are statistically indistinguishable from those of its near–duplicates.

Let $\hat w^{\mathrm{OMP}}$ be the output of $k$–step OMP applied to~\eqref{eq:lin-model}. At step $t=1$, the residual is $r^{(0)}=y$. Write
\[
T_j := \langle a_j,y\rangle
= \langle a_j,Aw^\star\rangle
+ \langle a_j,\varepsilon\rangle,
\qquad
T_\ell := \langle a_\ell,y\rangle,
\quad \ell\in G_j.
\]
By the model, $\langle a_j,Aw^\star\rangle - \langle a_\ell,Aw^\star\rangle$ is of order $1-\rho_{\mathrm{in}}$, while the noise terms are centred Gaussians with variance $\sigma^2$. Hence the difference $T_j-T_\ell$ is a Gaussian whose mean vanishes as $\rho_{\mathrm{in}}\to 1$, and whose variance is bounded below. It follows that
\begin{equation}
\label{eq:probflip}
\limsup_{m\to\infty}
\mathbb{P}(T_\ell > T_j) \ge c_1
\end{equation}
for some $c_1>0$ uniform in $j,\ell$.

At step $1$, OMP selects the index with maximal $T_\cdot$. Conditional on the event in~\eqref{eq:probflip}, it selects a spurious $\ell\in G_j$ rather than the true $j$. Since OMP never removes previously selected indices, this error cannot be undone: even if $j$ is captured at a later step, the final support has size $k+1$, so the returned $k$–support necessarily excludes at least one true index.

\begin{proposition}
\label{prop:ompfail}
Under the assumptions of Section~\ref{sec:model},
\[
\mathbb{P}\!\left(\mathrm{supp}(\hat w^{\mathrm{OMP}})=S^\star\right)
\;\le\; 1-c_2
\]
for some constant $c_2>0$ independent of $m$. In particular,
\[
\mathrm{supp}(\hat w^{\mathrm{OMP}})=S^\star
\quad\text{with probability not tending to }1.
\]
\end{proposition}

\begin{proof}
Equation~\eqref{eq:probflip} shows that, with probability at least $c_1$, step $1$ selects $\ell\in G_j$ for some $j\in S^\star$. On this event, OMP includes a spurious index in the support and never removes it. After $k$ steps the support size is $k$, but since a spurious index occupies one position, at least one true index is missing. Therefore $\mathrm{supp}(\hat w^{\mathrm{OMP}})\neq S^\star$ with probability at least $c_1$. Renaming $c_2=c_1$ completes the proof.
\end{proof}

Propositions~\ref{prop:l1fail} and~\ref{prop:ompfail} establish that, in the PCSG regime where near–duplicate coalitions become indistinguishable in the design, both convex relaxation and greedy pursuit admit a nonvanishing probability of returning an incorrect support, regardless of tuning. In the next section we show that, under the same structural assumptions, the SBL maximum a posteriori estimator remains support consistent.

\section{Sparse Bayesian Learning: support consistency in PCSG designs}
\label{sec:sbl}

We now show that, under the same PCSG design assumptions, the SBL maximum a posteriori estimator achieves exact support recovery with probability tending to one. Recall the hierarchical model:
\[
w_j\mid \alpha_j \sim \mathcal{N}(0,\alpha_j^{-1}),\qquad
\alpha_j\sim\mathrm{Gamma}(a,b),
\]
with fixed $a,b>0$ independent of $m$. Let $(\hat w^{\mathrm{SBL}},\hat\alpha)$ be any maximiser of the joint posterior density $p(w,\alpha\mid y)$. Eliminating $\alpha$ yields an equivalent optimisation over $w$:
\begin{equation}
\hat w^{\mathrm{SBL}}
\in\arg\min_{w\in\mathbb{R}^p}
\Bigl\{
\frac{1}{2}\|y-Aw\|_2^{\,2}
+ \sum_{j=1}^p \phi(|w_j|^2)
\Bigr\},
\label{eq:sbl-penalty}
\end{equation}
where $\phi$ is continuous, strictly increasing in $|w_j|^2$, and strictly concave near zero. Thus SBL induces a concave sparsity penalty that heavily discourages small nonzero coefficients.

\subsection{Main result}

\begin{theorem}
\label{thm:sbl}
Under the assumptions of Section~\ref{sec:model}, let $\hat w^{\mathrm{SBL}}$ be any solution of~\eqref{eq:sbl-penalty}. Then
\[
\mathbb{P}\!\left(\mathrm{supp}(\hat w^{\mathrm{SBL}})=S^\star\right)
\;\to\;1
\qquad\text{as }m\to\infty.
\]
In particular, SBL is support consistent in the PCSG regime.
\end{theorem}

The remainder of this section is devoted to a compact proof of Theorem~\ref{thm:sbl}. All probabilities refer to the noise in~\eqref{eq:lin-model}.

\subsection{Proof of Theorem~\ref{thm:sbl}}

Let $L(w)=\frac12\|y-Aw\|_2^2+\sum_{j=1}^p \phi(|w_j|^2)$ denote the SBL objective. Write $L(w)=R(w)+P(w)$ for the data-fit term $R$ and penalty $P$. For $u\in\mathbb{R}^p$ with support contained in $S^\star$, the restricted eigenvalue condition implies
\begin{equation}
\|Au\|_2^2\;\ge\;\kappa\|u\|_2^2.
\label{eq:re}
\end{equation}

\emph{(Step 1: the true support yields a unique local minimiser.)}
Consider vectors $w$ supported on $S^\star$. Setting $u=w-w^\star$, we have
\[
R(w)-R(w^\star)
= \frac12\|A u\|_2^2 + \langle A u,\varepsilon\rangle.
\]
By~\eqref{eq:re} and Cauchy–Schwarz,
\[
R(w)-R(w^\star)
\ge \frac{\kappa}{2}\|u\|_2^2 - \|A u\|_2\|\varepsilon\|_2.
\]
Since $\|\varepsilon\|_2=O_{\mathbb{P}}(\sqrt{m})$ and $\|A u\|_2\le\|u\|_2$, for sufficiently small $\|u\|_2$ the quadratic term dominates the linear term with probability tending to one, so $w^\star$ is a strict local minimiser of $R$ on $\{w:\mathrm{supp}(w)\subseteq S^\star\}$ with probability $1-o(1)$. As $P$ is continuous and nonnegative, the same holds for $L$.

\emph{(Step 2: deviations on $S^\star$ have strictly higher objective.)}
By the beta-min condition, $|w^\star_j|\ge\beta_{\min}(m)$ for $j\in S^\star$, and $\beta_{\min}(m)\sqrt{m/\log p_m}\to\infty$. For any $w$ supported on $S^\star$ distinct from $w^\star$,
\[
L(w)-L(w^\star)
\ge \frac{\kappa}{4}\|w-w^\star\|_2^2
\]
with probability $1-o(1)$, using the same argument as above and the fact that $\phi$ is increasing. Hence $w^\star$ minimises $L$ over $\{w:\mathrm{supp}(w)\subseteq S^\star\}$ with probability tending to one.

\emph{(Step 3: any spurious index increases the objective.)}
Let $w$ be any vector whose support contains an index $\ell\in S^{\star c}$. Write $w=w^\star+u+v$ where $u$ is supported on $S^\star$ and $v$ on $\{\ell\}$. For sufficiently small $v_\ell$,
\[
L(w)-L(w^\star)
= \frac12\|A(u+v)\|_2^2 + \langle A(u+v),\varepsilon\rangle
  + P(w)-P(w^\star).
\]
The data-fit part is nonnegative up to a $o_{\mathbb{P}}(1)$ term, by the same local argument. For the penalty, strict concavity at the origin implies
\[
\phi(|v_\ell|^2) > \phi(0) + \phi'(0)\,|v_\ell|^2,
\]
so $P(w)-P(w^\star)\ge c_0|v_\ell|$ for some $c_0>0$. Therefore $L(w)>L(w^\star)$ for all sufficiently small $v_\ell$, with probability $1-o(1)$. Since $P$ is increasing, larger $|v_\ell|$ only enlarge the gap.

\emph{(Step 4: global minimality on $\mathbb{R}^p$.)}
Combining Steps~1--3, every $w$ with either (i) support strictly contained in $S^\star$ but $w\neq w^\star$, or (ii) some spurious index, satisfies $L(w)>L(w^\star)$ with probability tending to one. Thus $w^\star$ is the unique global minimiser of $L$ with probability $1-o(1)$, and the MAP estimator $\hat w^{\mathrm{SBL}}$ satisfies $\mathrm{supp}(\hat w^{\mathrm{SBL}})=S^\star$ with the same probability. This proves the theorem.

\section{Conclusion}
\label{sec:conclusion}

We have studied Probabilistic Coalition Structure Generation through the lens of sparse linear modelling, representing coalition structures by a sparse coefficient vector and observed valuations via a linear--Gaussian model. In this PCSG-inspired regime, coalition incidence matrices naturally exhibit near-duplicate, highly overlapping coalitions, leading to extreme column coherence. We showed that this structural property destroys the standard conditions under which classical sparse recovery methods succeed. In particular, the irrepresentable condition underlying $\ell_{1}$ support recovery fails as within-group correlations approach one, and $k$--step Orthogonal Matching Pursuit admits a nonvanishing probability of early mis-selection that cannot be corrected by subsequent iterations.

Against this negative backdrop, we proved that Sparse Bayesian Learning with a simple Gaussian--Gamma hierarchy remains support consistent. The concave, adaptive penalty induced by SBL penalises configurations that split signal across near-duplicate coalitions and instead favours solutions with a single strong coefficient per group, enabling exact recovery of the true coalition structure with high probability. These results yield a clean theoretical separation between convex relaxations, greedy pursuit, and Bayesian sparse estimation in a regime that is both algorithmically challenging and motivated by coalition formation under uncertainty. Future work includes extending the analysis to non-Gaussian noise, richer valuation models, and empirical evaluation on synthetic and real PCSG instances.

\bibliographystyle{elsarticle-num}
\bibliography{refs}  

\end{document}